\documentclass[letterpaper, 10 pt, conference]{ieeeconf}

\IEEEoverridecommandlockouts                
\IEEEoverridecommandlockouts
\usepackage[utf8]{inputenc}
\usepackage[english]{babel}
\usepackage{amsmath}
\usepackage{amsfonts}
\usepackage{amssymb}
\usepackage{makeidx}
\usepackage{graphicx}
\usepackage{cite}
\usepackage{textcomp}
\usepackage{xcolor}
\usepackage{lmodern}
\usepackage{fourier}
\usepackage{algorithmic}
\usepackage{float}
\def\BibTeX{{\rm B\kern-.05em{\sc i\kern-.025em b}\kern-.08em
    T\kern-.1667em\lower.7ex\hbox{E}\kern-.125emX}}

\newtheorem{theorem}{Theorem}
\newtheorem{lemma}{Lemma}
\newtheorem{definition}{Definition}
\DeclareMathOperator{\sgn}{sgn}

\begin{document}

\title{\LARGE \bf Bipartite Consensus in the Presence of Denial of Service Adversary}

\author{Gopika R$^{1}$, Ankita Sharma$^{2}$ and Rakesh R Warier$^{3}$%

\thanks{$^{1}$ Gopika R. is a PhD student in the Department of Electrical, Electronics and Instrumentation engineering, BITS Pilani, KK Birla Goa Campus, Zuari Nagar, Goa, India
{\tt\small p20190058@goa.bits-pilani.ac.in}
}

\thanks{$^{2}$ Ankita Sharma is an Assistant Professor in School of Basic Science (Mathematics), Indian Institute of Information Technology, Una,  Himachal Pradesh, India 
{\tt\small ankita.iitm22@gmail.com}
}

\thanks{$^{3}$ Rakesh R Warier is an Assistant Professor in the Department of Electrical, Electronics and Instrumentation engineering, BITS Pilani, KK Birla Goa Campus, Zuari Nagar, Goa, India
{\tt\small rakeshw@goa.bits-pilani.ac.in}
}
}

\maketitle

\begin{abstract}

Attacks on a set of agents with cooperative and antagonistic interactions attempting to achieve linear bipartite consensus are considered here. In bipartite consensus, two clusters are formed, agents in each cluster converging to a final state which is negative of the other cluster's final state. The adversary seeks to slow down the bipartite consensus by a Denial of Service (DoS) type attack where the attacker has the capability to break a specific number of links at each time instant. The problem is formulated as an optimal control problem and the optimal strategy for the adversary is determined.\\
\end{abstract}


\section{INTRODUCTION}
Networked systems form the backbone of critical infrastructures such as power grids, transportation systems, home automation systems. Other examples of interconnected networked systems include robotic systems and autonomous aerial systems. These are interconnected by design separated geographically and functionally. Distributed coordination of multiple robotic vehicles, including unmanned aerial vehicles, unmanned ground vehicles, and unmanned underwater vehicles, has been a very active research subject by the systems and control community \cite{ren2005survey,cao2012overview}. The distributed control of multiple agents can offer more flexibility, scalability, adaptivity, and more robustness to individual failure. Additionally, they find applications in distributed networks including smart grids \cite{dorfler2013synchronization}, autonomous under water vehicles, unmanned aerial vehicles, etc.

Consensus problems normally deal with cooperative systems. In these types of cooperative systems, the consensus is achieved through the communication of agents which are characterized by edge weights that are non-negative. There can be situations wherein the agents in a system need to repel each other or maintain a certain formation in which we will be forced to use a negative weight for interaction. Therefore, in short, non-negative weights corresponds to a system with cooperative interactions and negative weights implies antagonistic interactions. The formation of several clusters is generally referred to as a multi-partite consensus. If the number of clusters in a multi-partite consensus is two, then they are referred to as bipartite consensus where the agent forms two different clusters and they converge to a point which is different only in the sign of the value \cite{altafini2012consensus, altafini2012dynamics}. We consider the problem of a set of agents trying to achieve bipartite consensus in the presence of an adversary.  

The networked multi-agent system can suffer from malicious interference which can disrupt the performance of the system. The robust performance of a networked multi-agent system requires physical attacks and cyber-attacks to be detected and countermeasures to be taken. Several attacks on networked control systems have been reported in the near past. These attacks can be classified as denial-of-service attacks, replay attacks, and fault data injection attacks \cite{teixeira2015secure}. The authors in \cite{dibaji2019systems} overview the existence of several categories of attacks on power and transportation fields on a control-theoretic framework. The study was conducted on the basis of prevention, resilience, and detection $\&$ isolation of the different types of attacks. Authors in \cite{amin2009safe} studied the problem of security on a linear dynamical system that is prone to DoS attacks and design a controller that minimizes the effect of this attack. A more comprehensive attack and system model considering fault detector mechanism and model uncertainties were proposed by \cite{teixeira2015secure}. The work by Pasqualetti et. al. \cite{pasqualetti2015control} explored the challenges of detection and identification of attack, using a descriptor system model. 

Preliminary work in the area of linear consensus of multi-agent systems under DoS attack was considered in \cite{khanafer2012consensus}. They model the problem with cooperative interactions among the agents and finds the optimal strategy to slow down the consensus. The extension of the work was made in \cite{khanafer2013robust} by introducing the effects caused by the presence of a defender. This system was also modeled based on Pontryagin's Maximum Principle and discussed the existence of saddle point equilibrium. In \cite{khanafer2013robust} the authors considered only the cooperative interactions among the agents. 

To the best of our knowledge, the performance of a network with cooperative and antagonistic interactions in the presence of an adversary remains unexplored. Here, we consider a set of agents that are interconnected by a signed graph that is structurally balanced. The set of agents are trying to achieve bipartite consensus. This system is disrupted by an external adversary which can cut links between agents. There is a constraint on the maximum number of links that can be broken at a certain time instant. The adversary has the ability to break both positive and negative links. The problem is formulated as an optimal problem and the optimal attack strategy for the adversary for slowing down the bipartite consensus is derived. The results are illustrated by numerical simulations. 

The rest of the paper is organized as follows, section \ref{sec:bipartite} talks about the basic preliminaries and section \ref{problemform} is about the problem formulation used in this article. Section \ref{sec:simulation} presents the simulation results with a numerical example and the article is concluded by section \ref{sec:conclusion} which also provides the future scope of the work.

\section{BIPARTITE CONSENSUS} \label{sec:bipartite}

When a group of agents interacts over a signed graph, they can form a bipartite consensus where values of agents converge to two clusters. In this section we summarize some fundamental results on bipartite consensus from works \cite{altafini2012consensus} and \cite{altafini2012dynamics}.

Consider a signed graph $G=\left(V, E, A \right)$, where $V=\left\{v_1, v_2, \ldots, v_n\right\}$ is the set of nodes or agents, $E$ is the set of edges and $A$ is the adjacency matrix that contains or represents the interconnection between the agents. The non-diagonal elements in the matrix $A$ represent the weights between the agents and the diagonal entries in the matrix are considered to be zero. The Laplacian matrix for the system can be defined as $L^G=C-A$, where $C$ is the connectivity matrix which will be a diagonal matrix with the values as the sum of the absolute value of the adjacent links connected to that particular node. Therefore, the Laplacian matrix $L^G$ can be represented as
\begin{eqnarray}
    l^{G}_{ik} = \left\{
    \begin{array}{ll}
        \sum_{j \rightsquigarrow i}\lvert{a_{ij}}\rvert  & k\ =\ i \\ 
        -a_{ij} & k\ \neq\ i.\\
    \end{array}
    \right. \label{eq:Laplacian}
\end{eqnarray}

The dynamics of $n$ agents $x_i\in \mathbb{R}$ connected by the graph $G$ is given by,
\begin{align}
    \dot{x}=-L^{G}x(t)
    \label{MathDynamics}
\end{align}
where $x =[x_1, \cdots , x_n]^\top$ and $L_G$ is calculated as in \eqref{eq:Laplacian}. The idea of structurally balanced graph is defined below. 

\begin{definition}
\label{strbalanced} \cite{altafini2012consensus}
A graph $G$ is structurally balanced if and only if 
\begin{enumerate}
    \item All cycles of $G$ are positive
    \item $0$ is an eigenvalue of $L^G$ which is the Laplacian matrix of the graph $G$ as defined in \eqref{eq:Laplacian}.
\end{enumerate}
\end{definition}

A structurally balanced graph that satisfies the properties as in definition \ref{strbalanced} corresponds to the partitioning of a signed graph into two sub-communities in such a manner all the edges within a sub-community have positive weights and all the edges connecting these sub-communities with other sub-communities has a negative weight. 

\begin{lemma}
The system \eqref{MathDynamics} admits a bipartite consensus solution if and only if the graph $G$ is structurally balanced \cite{altafini2012consensus}.

Here, bipartite consensus implies $\lim_{t \to \infty} x_i = | \beta | $ where $\beta \neq 0$, for all $i= 1, \cdots, n$.
\label{lemma structurally balanced}
\end{lemma}
\begin{proof}
The proof is provided in \cite{altafini2012consensus} and is omitted here for brevity. 
\end{proof}

Consider the example in fig \ref{fig1}. The system has $11$ agents that form two different groups. The agents within a sub-community have positive weights thus showing cooperative interaction and the weights between the agents of sub-community A and sub-community B are negative thus showing an antagonistic interaction. As shown in \cite{altafini2012consensus}, the agents of two sub-communities 'repel' each other, while moving towards agents of the same sub-community forming a bipartite consensus, essentially two clusters. 
\begin{figure*}[h] 
    \centering
    \includegraphics[scale=1]{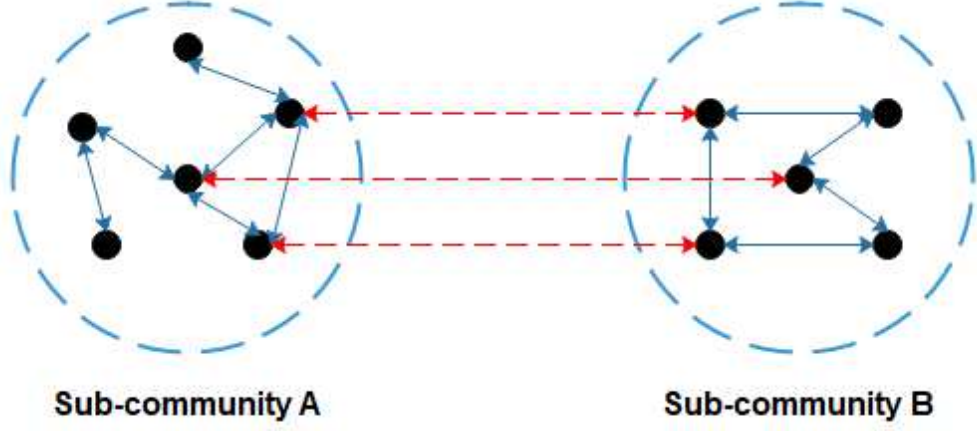}
    \caption{Bipartite consensus of a multi-agent system; The black circles indicate the different agents, blue connections (links) represent the cooperative interactions and the red connections indicate the antagonistic interaction }
    \label{fig1}
\end{figure*}

In the next section we consider the problem of bipartite consensus under attack.

\section{PROBLEM FORMULATION}
\label{problemform}
Consider a network of interconnected systems of weighted undirected graphs with $n$ agents or nodes. This can be represented using  $G=(V,E)$ where $V$ represents the vertices which are agents $\left[v_{1}, v_{2},..., v_{n}\right]$ and $E$ represents the edges. It is important that the network graph must be structurally balanced as in definition \ref{strbalanced}. It is assumed that the adversary has complete knowledge of the system. 
  
The state of the nodes at time $t$ is represented as $x=\left[x_{1}(t), x_{2}(t),...,x_{n}(t)\right]^{T}$. Let the dynamics of the system be 
\begin{align}
    \dot{x}=-L^{G}(t)x(t)
    \label{dynamics}
\end{align}
where $L^{G}$ is defined as the difference between the connectivity matrix $C$ and the adjacency matrix $A$. The matrix $A$ contains the weights of that particular links which are any scalar values and the self loops such as $a_{11} $ or $a_{22}$, etc are not considered. These values are considered in the connectivity matrix $C$. Therefore, as a result the matrix $C$ will be a diagonal matrix with the values being the sum of the absolute values of adjacent links connected to that node, $C_{ii}=\sum_{j \in adj(i)}\lvert{a_{ij}}\rvert$. The weights of the system can be either zero, positive value or a negative value based on the following condition wherein $C_k$ and $C_l$ represents the set of nodes that belongs to the two different clusters such that $C_k \cup C_l = V$ and $C_k\ \cap\ C_l\ =\ \emptyset$. $a_{ij}$ follows the standard rules if $(i,j) \in E$ and if $(i,j)$ belong to the same cluster, then $a_{ij}>0$ and if $(i,j)$ belongs to different cluster, then $a_{ij}<0$. $a_{ij}=0$ if $(i,j) \notin E$. From Definition \ref{strbalanced} when the graph is structurally balanced, the agents of the same cluster will have positive connections and the connection between the inter-cluster agents will be negative. From Lemma \ref{lemma structurally balanced}, this will result in agents achieving bipartite consensus, where each cluster of agents converging to values differing only in sign. Let us consider a structurally balanced network graph with the system dynamics as 
\begin{equation}
    \dot{x}=-L^{G}(t)x(t)
\end{equation}
where $L^{G}=C-A$.
\begin{eqnarray}
    l^{G}_{ik} = \left\{
    \begin{array}{ll}
        \sum_{j \rightsquigarrow i}\lvert{a_{ij}}\rvert  & k\ =\ i \\ 
        -a_{ij} & k\ \neq\ i.\\
    \end{array}
    \right.
\end{eqnarray}
where $j \rightsquigarrow i$ implies $j$ is connected to $i$. 

$L^{G}$ is a matrix that represents the system before the attack. The adversary aims to slow down the consensus by breaking the links that connect different agents. Let $\alpha$ be the maximum number of links that the adversary can disrupt at each time instant. The control is given by $u(t)=[u_{12}(t), u_{23}(t),...,u_{(n-1)n}(t)]^{T}$, where $u_{ij}(t) \in \{0, 1\}$. When $u_{ij}(t) =1$, the link $(i,j)$ is broken at time $t$. The weight of each link $(i, j)$ can be redefined as $a^{a}_{ij}=a_{ij}(1-u_{ij})$ and the graph under attack can be modelled as 
\begin{eqnarray}
    l^{Ga}_{ik} = \left\{
    \begin{array}{ll}
        \sum_{j \rightsquigarrow i}\lvert{a_{ij}(1-u_{ij})}\rvert  & k\ =\ i \\ 
        -a_{ij}(1-u_{ij})& k\ \neq\ i. \\
    \end{array}
    \right.
\end{eqnarray}
The objective function of the adversary can be defined as
\begin{equation}
    J(u)=\int ^{T}_{0}k(t)\ x(t)^{T}M(t)\ x(t)\ dt,
\end{equation}
where $k(t)$ is the positive kernel integrable over $[0,T]$ and $M$ is the connection matrix which indicates all the connections between the sub-communities of the bipartite network graph. 
\begin{align}
    M_{ij}=\left\{ 
    \begin{array}{cll}
    -\sgn(a_{ij}) & ; & i \ \neq \ j\\
    n-1 & ; & i \ = \ j  \\ 
    \end{array}\right.
\end{align}
The problem can be formulated in a optimal control framework as follows
\begin{align}
    & \underset{u(t) \in U}{\max}\ J(u) \nonumber \\
    s.t.\ \ \ & \dot{x}(t)\ =\ -L^{G}(t)x(t). 
\end{align}
The Hamiltonian is defined as
\begin{eqnarray}
    H(x, p, u) = k(t)x(t)^{T}Mx(t)-p(t)^{T}L^{G}(t)x(t) \label{Hamiltonian}
\end{eqnarray}
Considering the first order necessary conditions for optimality \cite{naidu2002optimal} can be formulated wherein the state and co-state dynamics becomes 
\begin{align}
    \frac{\partial H}{\partial p}  & = \dot{x}(t) = -L_{G}(t)x(t) \label{eq:opti1} \\
    \frac{-\partial H}{\partial x} & = \dot{p}(t) = -k(t)Mx(t)+L_{G}(t)^{T}p(t) \label{eq: opti2}\\
    \frac{\partial H}{\partial u} & = 0 \label{eqq:opti3}
\end{align}
The optimal control can be obtained as 
\begin{align}
u^{*} & = \underset{u \in U}{\arg \max}\ H(x,p,u) \label{opti:4}
\end{align}
To solve for $u^{*}$, we need to maximize the $H$. Solving for $p^{T}L^{G}x$ and substituting in (\ref{Hamiltonian})
\begin{eqnarray}
    p^{T}L^{G}x & = & \sum_{j=2}^{n}\sum_{i=1}^{n-1}\lvert{a_{ij}}\rvert(1-u_{ij})f_{ij}
\end{eqnarray}
 where $f_{ij}=\big(p_{i}-\sgn(a_{ij})p_{j}\big)\big(x_{i}-x_{j}\big)$.
\begin{equation}
     \max H(x, p, u)=k(t)x(t)^{T}Mx(t) - \sum_{j=2}^{n}\sum_{i=1}^{n-1}(1-u_{ij})f_{ij}
\end{equation}
\begin{eqnarray}
    u^{*}_{ij} = \left\{
    \begin{array}{cll}
        \{0, 1\} & ; & \text{if}\ f_{ij} = 0 \\
        1 & ; & \text{if}\ (i,j)\ \in I_{t} \\
        0 & ; & \text{if}\ f_{ij}>0.
    \end{array}
    \right.
\end{eqnarray}
 where $I_{t}$ is the ordered set of values of $f_{ij}$ that are less than zero and can be attacked, that is, 
 \begin{align}
     I_{t} & = \big\{ (i,j) \in V : f_{ij}<0,  \text{ and } f_{ij}<f_{\alpha + 1} \big\}
 \end{align}
 
The function $f_{ij}$ depends on both the state $x(t)$ and co-state $p(t)$ variables. The objective of the adversary is to slow down the consensus and for this he can break the links that connect these agents with each other. Therefore, we need to find an optimal strategy for the adversary so that the $J(u)$ increases because of the disruption caused by him. We need to look at the cost function of the adversary to determine the optimal strategy which is explained using the following theorem.

\begin{theorem}
For a particular time $t$, the optimal strategy of the adversary is to break the $\alpha$ number of links with the highest $w_{ij}$ value where $w_{ij}=\lvert{a_{ij}}\rvert \big(x_{j}-\sgn(a_{ij})x_{i} \big)^{2}.$ 
\end{theorem}
\textbf{Proof.} Let us consider $\alpha$ be the number of links that are break by the adversary at time instant $t$. The value of $f_{ij}$ cannot change in a finite interval as it solely depends on $x(t)$ and $p(t)$ which does not change in that interval which results in the control to remain unchanged. 

Let $[s, s+\Delta s]$ be a sub-interval of $[0, T]$ in which the adversary applies the strategy $u^{A}$ with a system matrix as $L_{1}$. Therefore, the state trajectory is given as
\begin{align*}
    x(t)=e^{-L_{1}(t-s)}x(s) \ ;\ \ t \in [s, s+\Delta s]
\end{align*} 
Let $P(t)$ be $e^{-L_{1}t}$. Let $\nu$ be the eigenvector corresponding to eigenvalue $0$ of the matrix $L_{1}$. Let $\nu\nu^{T}$ be $K$. Then the connection matrix $M$ can be expressed as $ (nI-K) $  where $n$ is the number of agents in the network. Let $u^{B}$ be the strategy that is similar to the $u^{A}$ except at link $(i, j)$ with a system matrix as $L_{2}$. At $ (i, j) $, $u^{A}_{ij}=0$ where as $u^{B}_{ij}=1$. Let $Q(t)$ be a doubly stochastic matrix such that  $Q(t)=e^{-L_{2}t}$,\ $t\geq$0. We need to show that strategy $u^{B}$ is better than $u^{A}$. The aim of the adversary is to maximize the cost function. We can obtain the conditions for which the cost for $u^{A}$ greater than the cost for $u^{B}$. 
\begin{align}
& \int_{s}^{s+\Delta s} k(t)x(s)^{T}P(t-s)MP(t-s)x(s) dt    \nonumber \\
        & < \int_{s}^{t^{*}} k(t)x(s)^{T}P(t-s)MP(t-s)x(s) dt +\int_{t^{*}}^{s+\Delta s}k(t) \nonumber \nonumber \\
        & x(s)^{T}P(t^{*}-s)Q(t-t^{*})MQ(t-t^{*})P(t^{*}-s)x(s) dt. 
\end{align}

The cost of strategy corresponding to $u^{B}-u^{A}$ can be written as
\begin{equation*}
    \begin{array}{l}
        \int_{s}^{s+\Delta s} kx^{T}[P(t^{*}-s)Q(t-t^{*})MQ(t-t^{*})P(t^{*}-s)- \\
        P(t-s)MP(t-s)]x dt >0 \\
    \end{array}
\end{equation*}
\begin{equation}
    \int_{s}^{s+\Delta s} kx(s)^{T}\Lambda(t,t^{*}) x(s) dt > 0\ ,
\end{equation}
 
where $\Lambda=P(t^{*}-s)Q(t-t^{*})MQ(t-t^{*})P(t^{*}-s)- P(t-s) M P(t-s)$. As $t$ decreases, $P(t)=I+tL_{1}+\mathcal{O}(t^{2})$ and $Q(t)=I+tL_{2}+\mathcal{O}(t^{2})$. Substituting for $P(t), Q(t)$ and $M$, the $\Lambda(t,t^{*})$ can be simplified as,
\begin{equation}
    \Lambda(t,t^*)=-2n(t-t^*)(L_2-L_1)+\mathcal{O}(t^2).
    \label{lam}
\end{equation}
 
For sufficiently small $t$ and $t^*$, the first term is very much larger than the second term, therefore it dominates. At the link $(i, j)$, $L_{1_{ij}}>L_{2_{ij}}=0$, and at all the links $L_{1_{kl}}=L_{2_{kl}}$ where $(k, l) \neq (i, j)$.
\begin{equation}
    \begin{array}{lll}
        h(t,x(s)) & = & 2n(t-t^{*}) \times \\
        & & \sum_{l=1}^{n} \sum_{k=1}^{l-1} (L_{1kl}-L_{2kl})(x_{k}-\sgn(L_{ij})x_{l})^2  \\
        & = & 2n(t-t^{*})L_{1ij}(x_{i}-\sgn(L_{ij})x_{j})^2 
    \end{array}
    \label{eq:finalH}
\end{equation}

In a bipartite consensus, the adversary can either break the link with a positive weight or he can break the link with a negative weight. When the link with a positive weight is attacked, (\ref{lam}) becomes $h(t,x(t))=2n(t-t^{*})\lvert{a_{ij}}\rvert(x_{i}-x_{j})^{2}$ and if the link with a negative weight is attacked, then $h(t,x(t))=2n(t-t^{*})\lvert{a_{ij}}\rvert(x_{i}+x_{j})^{2}$. This is possible only when the graph has multiple connected components and at least one link exists between them. If the graph at time $s$ is a single connected component, then the breaking of the link by the adversary to slow down the convergence is not possible as the system has already reached bipartite consensus. This proof provides an easy way of calculating optimal strategy and picturing the DoS attack. In the cost $J(u)$, each term corresponds to the value of $w_{ij}$. Therefore, it is inferred that the optimal strategy for the attacker to break the bipartite consensus and thus increase or maximize the cost function is by breaking the link with the highest value of $w_{ij}$. The attacker has complete knowledge of the system and that he is able to disconnect a link for the entire process. It will be interesting to see how the solution would change if there was a constraint on the amount of time an attacker can break a link.
 
\section{SIMULATION RESULTS} \label{sec:simulation}
We consider a structurally balanced graph with $n=4$ as shown in fig \ref{fig3}(a) for illustration. 
\begin{figure}[h]
    \centering
    \includegraphics[scale=1.5]{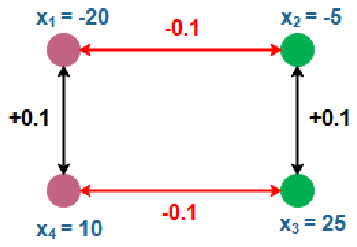}
    \includegraphics[scale=0.9]{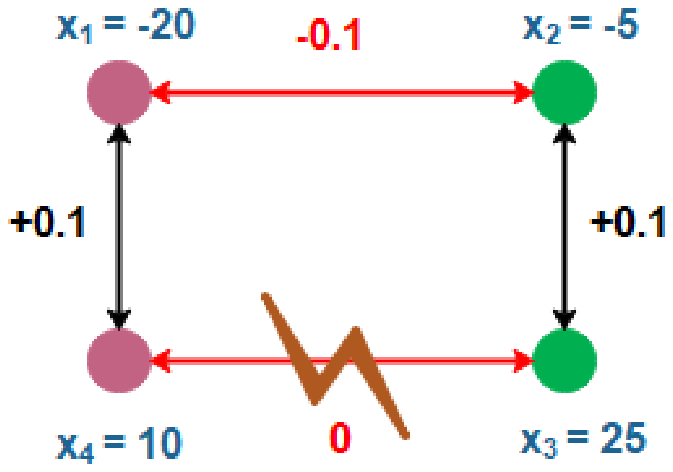}
    \caption{(a) Example of a four agent structurally balance graph with cooperative and antagonistic interactions; (b) The graph under attack of the adversary where the link $(3, 4)$ was attacked and therefore the weight of that particular link was changed to $0$.}
    \label{fig3}
\end{figure}
Therefore, according to the theory discussed above, the Laplacian matrix $L_1$ of the system can be written as, 
\begin{eqnarray*}
    L_{1}=\left[\begin{array}{cccc}
         0.2 & 0.1 & 0 & -0.1 \\
         0.1 & 0.2 & -0.1 & 0 \\
         0 & -0.1 & 0.2 & 0.1 \\
         -0.1 & 0 & 0.1 & 0.2 
    \end{array}\right]
\end{eqnarray*}
The initial conditions are chosen randomly between any interval. Here, let the initial conditions be $[-20, -5, 25, 10]^{T}$. This is the bipartite consensus where the four agents formed two different clusters. Now the adversary acts upon. In order for the adversary to give the best attack, he tries to attack that link that has the highest $w_{ij}$ value. Now, the $w_{ij}$ value is calculated for all the links. This value is found to be the highest for the link $(3, 4)$, with $w_{34}=122.5$. Therefore, the optimal strategy of the adversary to slow the bipartite consensus is to break this link at that particular time. Therefore, the $L_{1}$ matrix changes to a new matrix $L_{2}$.
 \begin{figure}[h]
    \centering
    \includegraphics[scale=0.7]{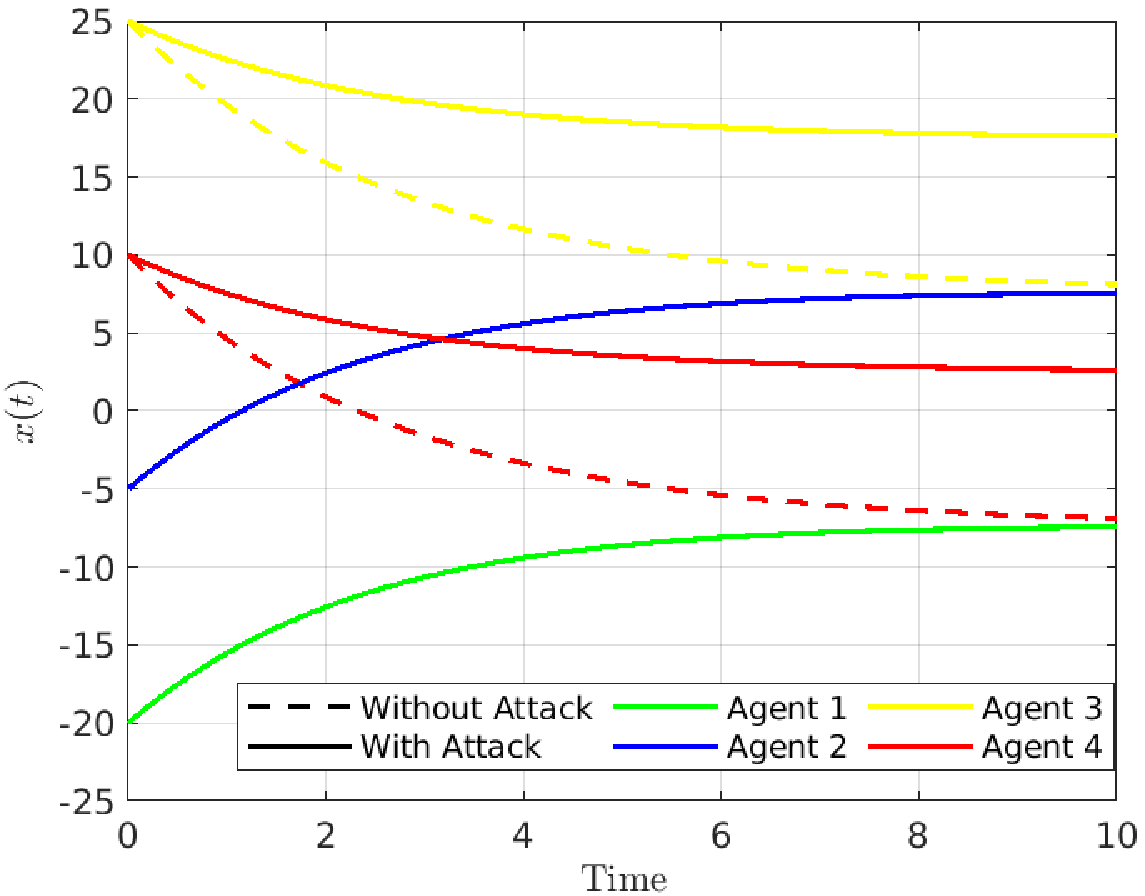}
    \includegraphics[scale=0.7]{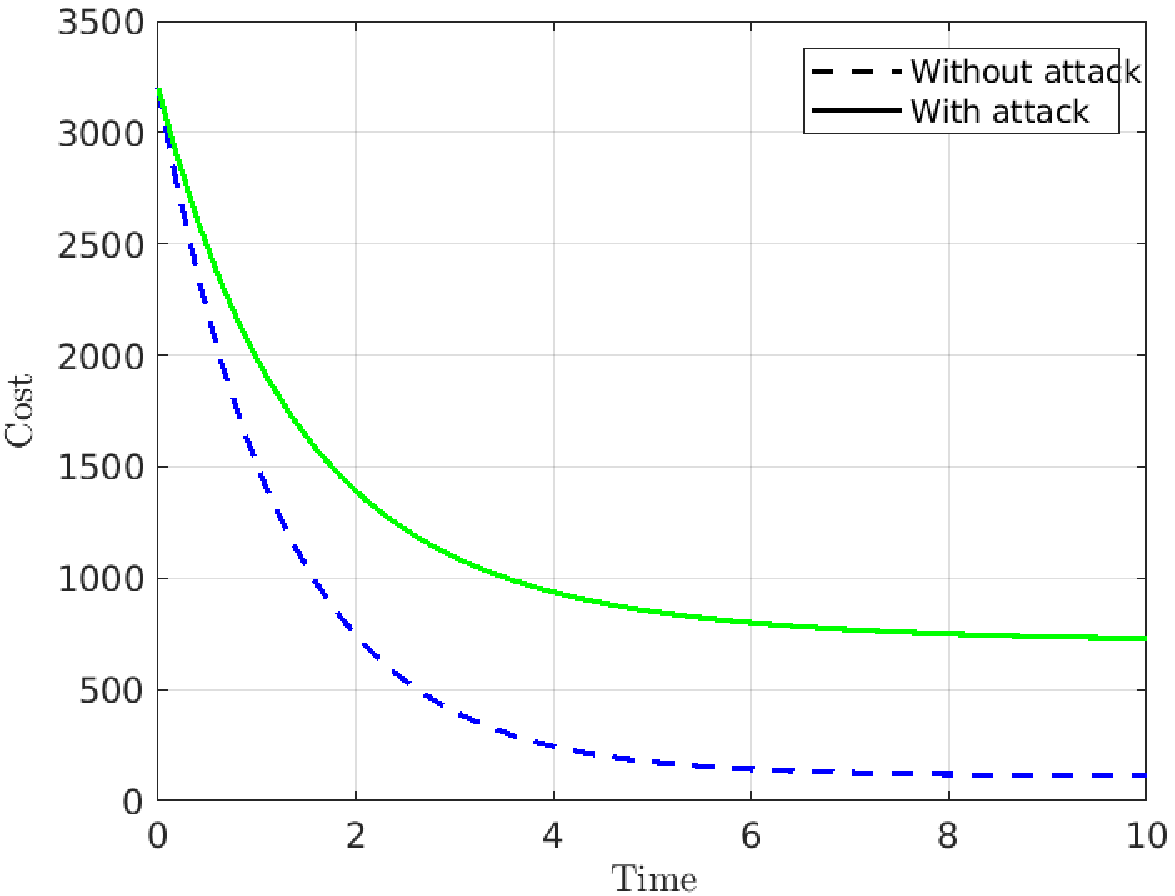}
    \caption{(a) Simulation result of the graph before and after the attack by the adversary; (b) The cost of the bipartite consensus $(J)$ with and without the attack}
    \label{fig4}       
\end{figure} 
From fig \ref{fig4}(a), we can infer that before the presence of an adversary the four agent network graph described in the illustration moves to bipartite consensus. These converging values are only different in sign. The speed of convergence can be adjusted by changing the weight of each link. After the attack by the adversary at the link $(3, 4)$, the bipartite consensus is slowed. From the simulation results, it is clear that the agents $1$ and $2$ still moves to the said consensus point whereas the adversary fails the bipartite consensus by breaking the link $(3, 4)$. The agents $3$ and $4$ deviates from the original trajectory. From \ref{fig4}(b), it is clear that the cost of the bipartite consensus $J(u)$ before the presence of the attack is nearly zero but when the adversary acts upon the graph by breaking the link $(3, 4)$, the cost increases.  The cost of bipartite consensus is thus maximized, slowing down the bipartite consensus.
 
\section{CONCLUSION} \label{sec:conclusion}
When a set of agents interact with each other using both cooperative and antagonistic interactions, and the interaction graph is structurally balanced, the agents end up with a bipartite consensus. The effect of the denial of service attack on a bipartite consensus problem was studied here. The optimal plan of action that can be followed by an adversary at the time of the attack on a structurally balanced multi-agent system was found out. It has been demonstrated by MATLAB simulation results. The future work will explore the case of having an attacker with incomplete knowledge and the additional constraint on time. 

\bibliographystyle{ieeetran}
\bibliography{main}

\end{document}